\documentclass{article}
\usepackage{geometry}
\usepackage[usenames,dvipsnames]{xcolor}
\usepackage{tikz}

\usepackage{amssymb,amsfonts,amsthm,amsmath,thmtools}
\usepackage{enumerate}
\usepackage{mathrsfs,mathtools}
\usepackage{titlesec}
\usepackage{booktabs, tabularx} 
\usepackage{graphicx} 
\usepackage{bbm}
\usepackage[hidelinks]{hyperref}
\hypersetup{
  colorlinks,
  citecolor=Violet,
  linkcolor=Black,
  urlcolor=Blue}

\usepackage{cleveref}
\usepackage{wrapfig}
\usepackage{changepage}

\def \cA {\mathcal{A}}
\def \cB {\mathcal{B}}
\def \cC {\mathcal{C}}
\def \cD {\mathcal{D}}

\def \cF {\mathcal{F}}
\def \cG {\mathcal{G}}

\def \cL {\mathcal{L}}
\def \cM {\mathcal{M}}

\def \cQ {\mathcal{Q}}
\def \cR {\mathcal{R}}
\def \cS {\mathcal{S}}

\def \bE {\mathbb{E}}

\def \bN {\mathbb{N}}

\def \bP {\mathbb{P}}
\def \P {\mathbb{P}}
\def \bQ {\mathbb{Q}}
\def \R {\mathbb{R}}

\def \bT {\mathbb{T}}

\def\cone{\operatorname{cone}}
\def\conv{\operatorname{conv}}

\def\AVaR{\operatorname{AVaR}}

\DeclareMathOperator*{\esssup}{ess\,sup}
\DeclareMathOperator*{\essinf}{ess\,inf}

\def \half {\frac{1}{2}}

\let\phi\varphi
\let\epsilon\varepsilon

\usepackage{xspace}

\newcommand{\one}[1]{\mathbbm{1}_{#1}}
\newcommand{\ind}[1]{\mathbbm{1}_{\left\{{#1} \right\}}}

\newcommand{\condE}[2]{\mathbb{E}\left[ \left. #1 \, \right| #2 \right]}

\newcommand{\V}{\mathbf{V}}

\newcommand{\wt}[1]{ \widetilde{ #1 }}

\newcommand{\for}{ \text{ for }}
\newcommand{\q}[1]{\quad \text{ #1 } \quad}
\newcommand{\qq}[1]{\qquad \text{ #1 } \qquad}
\def\qiq{\quad \implies \quad}

\newcommand{\ul}[1]{{#1}}
\newcommand{\ol}[1]{\overline{#1}}

\renewcommand \emptyset \varnothing

\newtheorem{thm}{Theorem}[section]
\newtheorem{theorem}[thm]{Theorem}

\newtheorem{prop}[thm]{Proposition}
\newtheorem{lem}[thm]{Lemma}

\newtheorem{claim}[thm]{Claim}

\theoremstyle{definition}
\newtheorem{defn}[thm]{Definition}

\newtheorem{eg}[thm]{Example}

\theoremstyle{remark}
\newtheorem{rmk}[thm]{Remark}

\renewcommand\thmcontinues[1]{Continued}

\newcommand{\cAn}{\cA_0}
\newcommand{\uD}{\ul \cD}
\newcommand{\uB}{\ul \cB}
\newcommand{\ce}{\varepsilon}

\title{On representing and hedging claims for coherent risk measures}
\usepackage{authblk}

\author[1,2]{Saul Jacka\thanks{Saul D. Jacka gratefully acknowledges funding received from the EPSRC grant EP/P00377X/1 and is also grateful to the Alan Turing Institute for their financial support under the EPSRC 
grant EP/N510129/1. E-mail: \textit{s.d.jacka@warwick.ac.uk}\\Department of Statistics\\University of Warwick\\Coventry CV4 7AL, UK}}
\author[1]{Seb Armstrong \thanks{Seb Armstrong gratefully acknowledges funding received from the EPSRC Doctoral Training Partnerships grant EP/M508184/1. E-mail: \textit{seb.armstrong@gmail.com}\\Department of Statistics\\University of Warwick\\Coventry CV4 7AL, UK}}
\author[3]{Abdelkarem Berkaoui \thanks{E-mail: \textit{berkaoui@yahoo.fr}\\College of Sciences\\Al-Imam Mohammed Ibn Saud Islamic University\\P.O. Box 84880\\Riyadh 11681\\Saudi Arabia\\}}

\affil[1]{University of Warwick}
\affil[2]{Alan Turing Institute}
\affil[3]{Imam Muhammad ibn Saud Islamic University}

\date{}
\begin{document}
\maketitle

\begin{abstract}
 We provide a dual characterisation of the weak$^*$-closure of a finite sum of cones in $L^\infty$ adapted to a discrete time filtration $\cF_t$: the $t^{th}$ cone in the sum  contains bounded random variables that are $\cF_t$-measurable.  Hence we obtain a generalisation of Delbaen's m-stability condition \cite{D06} for the problem of reserving in a collection of num\'eraires $\V$, called $\V$-m-stability, provided these cones arise from acceptance sets of a dynamic coherent measure of risk \cite{artzner1997applebaum,ADEH99}. We also prove that $\V$-m-stability is equivalent to time-consistency when reserving in portfolios of $\V$, which is of particular interest to insurers.

\end{abstract}



\section{Introduction}

Insurers reserve for future financial risks  by investing in a suitably prudent asset. Reserving is done in a particular unit of account, typically cash, or any other asset universally agreed to always hold positive value. We call such assets num\'eraires, examples of which include paper assets, such as currencies, or physical commodities.  Reserving a sufficient amount ensures that the risk carried by the insurer is acceptable. In some circumstances, the choice of num\'eraire is  clear; in others, it is not,  for example insurers reserving for claims in multiple currencies. The sufficient amount to reserve is modelled by a coherent measure of risk.

Coherent risk measures were first introduced by Artzner, Delbaen, Eber and Heath \cite{artzner1997applebaum,ADEH99}, in order to give a broad axiomatic definition for monetary measures of risk.  Financial positions are modelled as essentially bounded random variables on a suitable probability space $(\Omega,\cF,\P)$. A coherent risk measure is a real-valued functional on $L^\infty (\Omega,\cF,\P)$ that is cash invariant, monotone, convex, and positive homogeneous; see \cite{FS11}.  A coherent risk measure assigns a real value to every financial position: those with non-positive risk are deemed acceptable. We denote by $\cA$ the set of acceptable claims. It is easily shown that $\cA$ is a cone in $L^\infty$.

A coherent risk measure is a reserving mechanism: we assume that an insurer is making a market in (or at least reserving for) risk according to a coherent risk measure $\rho$ and they charge for or reserve for a random claim $X$ the price $\rho (X)$. Thus the aggregate position of holding the risky claim $X$ and reserving adequately should always be acceptable to the insurer.

 A coherent risk measure $\rho$ satisfies the Fatou property if, for any unifformly bounded sequence $X^n$ converging to $X$ in probability, 
 \[\liminf_n \rho(X^n) \ge \rho(X).\]
 A coherent risk measure satisfies the Fatou property if and only if, for some set of probability measures $\cQ$ absolutely continuous with respect to $\P$, we can represent $\rho$ as 
\[ \rho(X) = \sup_{\bQ \in \cQ} \bE_\bQ[X].\]
Recall that the dual of $L^\infty(\Omega,\cF,\P)$ is the space of all  finitely additive measures on $(\Omega,\cF)$ that are absolutely continuous with respect to $\P$. The Fatou property allows us to restrict our search for dual optimisers to elements in $L^1(\Omega,\cF,\P)$, identified with probability measures through their Radon-Nikodym derivative. We equip the space $L^\infty$ with the weak$^*$ topology $\sigma(L^\infty, L^1)$, so the topological dual is $L^1$. The acceptance set $\cA$ is weak$^*$-closed.

We assume that the insurer can trade at finitely many times $\{0,1,\dots,T\}$. At each time $t$, the insurer can re-evaluate the risk, conditional on the information in the sigma algebra $\cF_t$. A conditional coherent risk measure is the natural generalisation of a coherent risk measure; again,  such a measure $\rho_t$ satisfies the Fatou property if and only if, for a set $\cQ_t$ of $\P$-absolutely continuous probability measures we may represent $\rho_t$ by
\[ \rho_t(X)=\esssup_{ \bQ \in \cQ_t} \bE_{ \bQ}[X| \cF_t] .\]
In what follows, we fix $\cQ_t = \cQ$ for all $t$, and define the time $t$-acceptance set, $\cA_t$, to be the set of all claims $X\in L^\infty(\Omega,\cF,\P)$ with $\rho_t(X) \le 0$.

The simplest act of reserving  is to hold a set amount of cash $\rho(X)$ until the insurer must pay the claim $X$. More generally, starting with an amount $\rho(X)$ of cash, an insurer trades in any financial asset available, constructing a self-financing strategy with a terminal value equal to or exceeding the value of the claim $X$ at maturity. If this strategy is built by trading in the set of assets $\mathbf{V}=(v^0,\dots,v^d)$ as num\'eraires, then we shall say that the claim may be \emph{represented} by the vector $\mathbf{V}$. If we allow ourselves a large enough collection of assets, then representation is always possible: to hedge the bounded claim $X$ we need only buy and hold a claim whose value is $X$. Interest, therefore, should be focused on choosing a parsimonious collection of representing num\'eraires $\mathbf{V}$, and in identifying when such a collection is representing.

We will define a claim $X$ to be predictably representable by $\V$ if, starting from a reserve $\rho(X)$, we may transfer risk through each time period by trading in $\mathbf{V}$ in an acceptable manner, such that the terminal wealth equals the value of the claim: for portfolios $Y_t \in L^\infty(\Omega,\cF_t,\P;\R^{d+1})$, we have
\[ X = \rho_0(X)  + \sum_{t=0}^{T-1} (Y_{t+1} - Y_t) \cdot\mathbf{V}, \]
where each increment satisfies $\rho_t((Y_{t+1} - Y_t) \cdot\mathbf{V}) \le 0$. We write $\cA_t(\mathbf{V})$ for the set of all portfolios in $\mathbf{V}$ that are time-$t$ acceptable. The acceptance set $\cA_0$ is predictably $\mathbf{V}$-representable if it is the weak$^*$ closure of the sum of the  cones $ K_t(\cA,\V):= \cA_t(\mathbf{V}) \cap L^{\infty}(\Omega,\cF_{t+1},\P; \R^{d+1})$, 
\[ \cA_0(\mathbf{V}) = \ol{ \oplus_{t=0}^{T-1}K_t(\cA,\V) }. \]
For $X$ to be \emph{predictably} representable, we mean that $X$ is attainable (representable) as a sum of claims $X=\sum_t C_t$, where each $C_t$ is realised over the time period $(t,t+1]$, and pays out at time $t+1$.  Or equivalently, every element of $\cA_0$ is attainable by a collection of one-period bets in units of $\mathbf{V}$ at times $0,1,\dots,T-1$, and trades at time $1,\dots,T$. 

A key contribution of this paper is to provide the dual characterisation of $\V$-representability. Recall Delbaen's multiplicative stability (henceforth m-stability) condition, on the set of probability measures $\cQ$. We identify probability measures in $\cQ$,  via their Radon-Nikodym derivative, with random variables in the dual cone 
\[\cA_0^* = \{ Z \in L^1: \bE[ZX] \le 0 \quad \forall X \in \cA_0\}.\]  
The dual cone $\cA_0^*$ is m-stable if, for any stopping time $\tau$ and  $Z_1, Z_2 \in \cA_0^*$ 
 such that $\condE{Z_1}{\cF_\tau} = \alpha\condE{Z_2}{\cF_\tau}$, then $\alpha Z_2\in \cA_0(\V)^*$. See \cite{D06}. Likewise, the dual cone $\cA_0(\V)^*$ is $\V$-m-stable if, for any stopping time $\tau$ and  $Z_1, Z_2 \in \cA_0(\V)^*$  such that $\condE{Z_1}{\cF_\tau} = \alpha\condE{Z_2}{\cF_\tau}$, then $\alpha Z_2\in \cA_0(\V)^*$.
To show the equivalence of $\V$-m-stability and $\V$-representability, we present an elegant dual of each summand in the representation $K_t(\ul\cA,\mathbf{V})^*=\cM_t(\ul\cA(\mathbf{V})^*)$, called the \emph{predictable pre-image} of $\cA_0(\mathbf{V})^*$ at time $t$. Aside from being  useful in proving the equivalence of $\mathbf{V}$-predictable representability and predictable $\mathbf{V}$-m-stability, the predictable pre-image of a predictably m-stable convex cone $\cA_0(\mathbf{V})^*$ at time $t$ is a concrete description of the dual of the set of portfolios  held at time $t$ in order to maintain an acceptable position until time $t+1$.

We prove that $\V$-representability is equivalent to time-consistency of the risk measure. A risk measure is time-consistent if $\rho_t = \rho_t \circ \rho_{t+1}$. That is, today's reserve for a claim $X$ is precisely enough to reserve for tomorrow's reserve for $X$; see \cite{gianin2006risk,D06,riedel2004dynamic,roorda2005coherent} for examples of such measures. The sequence $(\rho_t)$ is not necessarily  time-consistent; see for example \cite{boda2006time,CS09}. Considerations of time-consistency are important for banks modelling Risk-Weighted Assets (RWAs) under the Basel III accords. A recent consultative document \cite{BCBS} highlights the change in methodology from using risk measures based on Value at Risk (VaR) to those based on Expected Shortfall (ES), also known as Average Value at Risk (AVaR, see  \cite{EPRWB14}). As shown by Cheridito and Stadje \cite{CS09}, AVaR is not  time-consistent.

In section 2, we elaborate on our generalisations of the  three  properties: namely  $\mathbf{V}$-time-consistency, $\mathbf{V}$-representability, and $\mathbf{V}$-m-stability. Throughout the section we illustrate our definitions with a toy example of Average Value at Risk. The main result of this paper is the equivalence of the three properties. 

In section 3, we provide some examples. In section 4, we prove the main result. We highlight the role that the filtration $(\cF_t)_{t=0,\ldots , T}$ plays.

\section{Pricing measures}

We recall some definitions and concepts. We fix a terminal time $T\in \bN$, a discrete time set $\bT := \{0,1,\dots, T\} $. We fix a probability space $(\Omega,\cF, \P)$, where $\bP$ is the \emph{reference measure} or \emph{objective measure}. The filtration $(\cF_t)_{t\in \bT}$ describes the information available at each time point. The space of all $\P$-essentially bounded $\cF$-measurable random variables is $L^\infty = L^\infty(\Omega,\cF,\P)$; we abbreviate $ L^\infty(\Omega,\cF_t,\P)$ to $L^\infty_t$.  The space of essentially bounded $\R^{d}$-valued random variables is $\cL^\infty(\R^{d})= L^\infty(\Omega,\cF,\P;\R^d)$. We denote the cone of non-negative (respectively strictly positive) essentially bounded random variables by $L^\infty_+$ (respectively  $L^\infty_{++}$). We denote the canonical basis vectors in $\R^{d+1}$ by $e_0,\ldots, e_d$.

At each time $t \in \bT$, we wish to price monetary risks using all information available at that time. Recall the following definition, adapted from \cite{DS05}:
\begin{defn}\label{defn:ccrmp}
A map $\rho_t : L^\infty \to L^\infty_t$ for $t \in \bT$ is a \emph{conditional convex risk measure} if, for all $X, Y \in L^\infty$, it has the following properties:
\begin{itemize}
\item
Conditional cash invariance: for all $m \in L^\infty_t$,
\[ \rho_t (X+ m ) = \rho_t(X) + m\qquad  \P\text{-almost surely;}\]
\item
Monotonicity: if $X \le Y$ ${\P}$-almost surely, then $\rho_t(X) \le \rho_t(Y)$;
\item
Conditional convexity: for all $\lambda \in L^\infty_t$ with $0 \le \lambda \le 1$,
\[ \rho_t(\lambda X + (1- \lambda)Y) \le \lambda \rho_t(X) + (1-\lambda)\rho_t(Y) \qquad  \P\text{-almost surely;}\]
\item
Normalisation: $\rho_t(0)=0$ ${\P}$-almost surely.
\end{itemize}
Furthermore, a conditional convex risk measure is called \emph{coherent} if it also satisfies 
\begin{itemize}
\item
Conditional positive homogeneity: for all $\lambda\in L^\infty_t$ with $\lambda \ge 0$,
\[ \rho_t(\lambda X)= \lambda\rho _t (X) \qquad  \P\text{-almost surely.} \]
\end{itemize}
\end{defn}
Our interest lies chiefly in reserving for and pricing liabilities. We see a positive random variable $X$ as a gain, and a negative $X$ as a loss, which explains the choice of sign in the cash invariance property, and the direction of monotonicity.
\begin{defn}
A convex risk measure satisfies the \emph{Fatou property} if, for any bounded sequence $(X^n)_{n \ge 1} \subset L^\infty$ converging to $X\in L^\infty$ in probability, we have 
\[ \rho_t(X) \le \liminf_{n\to \infty} \rho_t(X^n).\]
\end{defn}
  The Fatou property is equivalent to \emph{continuity from above}: $\rho_t$ is continuous from above if, whenever $(X^n)_{n \ge 1} \subset L^\infty$ is a non-increasing sequence such that $X^n \downarrow X$ $\P$-a.s., then
\[ \rho_t(X^n) \downarrow \rho_t(X) \qquad \P\text{-a.s. as }n \to \infty \]

\begin{defn}
A \emph{dynamic coherent risk measure} is a collection $\rho = (\rho_t)_{t=0,\dots,T}$, where each $\rho_t$ is a conditional coherent risk measure satisfying the Fatou property with representing set of measures $\cQ$:
\[ \rho_t(X)=\esssup_{ \bQ \in \cQ} \bE_{ \bQ}[X| \cF_t] .\]
\end{defn}
The \emph{acceptance set} of a conditional coherent risk measure $\rho_t : L^\infty \to L^\infty_t$ is 
\[ \cA_t = \{ X \in L^\infty : \rho_t(X) \le 0\}. \]

For the following results, we refer the reader to \cite{FS11} and \cite{DS05}. We equip the space $ L^\infty$ with the weak$^*$ topology $\sigma( L^\infty,  L^1)$, so  that the topological dual will be $ L^1$. Recall that a set $\cC$ of claims is arbitrage-free whenever
\[ \cC \cap L^\infty_+ = \{ 0\}.\]
\begin{prop}\label{prop:background}
For each $t$, define $\cA_t$ to be the acceptance set of the dynamic conditional \emph{coherent} risk measure  $\rho_t : L^\infty \to L^\infty_t$ satisfying the Fatou property.

Then  $ \cA_t$ is a weak$^*$-closed\footnote{in $ L^\infty$, i.e., $\cA_t$ is closed in the topology $\sigma( L^\infty,  L^1)$} convex cone that is  stable under multiplication by bounded positive $ \cF_t$-measurable random variables, contains  $ L^\infty_- $, and is arbitrage-free.
\end{prop}
\begin{rmk}
From now on, to emphasis that a coherent risk measure is, in general conditional, we shall refer to the acceptance set $\cA_0$ rather than $\cA$
\end{rmk}

\paragraph{Num\'eraires} A \emph{num\'eraire} is defined to be a random variable $v \in L^\infty_{++}$ such that $1/v \in L^\infty_{++}$. 
We shall from here on fix a finite collection of num\'eraires $\mathbf{V}=(v^0, \dots, v^d)$, with $v^0 \equiv 1$.

\subsection{Time-consistency}

In this and the subsequent sections we identify the probability measures $\bQ$ of the set $\cQ$ with their Radon-Nikodym derivatives $\frac{d\bQ}{d\P}$. We trust that which version is to be used will be clear from the context. The following definition is taken from Acciaio et al. \cite{AFP12}.

\begin{defn}\label{def:stdTC}
A dynamic coherent risk measure for random variables $(\ul\rho_t)_{t\in \bT}$ is \emph{(strongly)  time-consistent} if for all $t \le T-1$, and for all $X\in L^\infty$,
\[ \rho_{t}(X) = \rho_{t}(\rho_{t+1}(X)) . \]
\end{defn}

We note that the reserve for $X$ at time $t$ is $\rho_t(X)$. The generalisation of strong time-consistency to $\V$-time-consistency is:

\begin{defn} A dynamic convex risk measure $\rho=(\rho_t)_{t=0,\dots,T}$ is \emph{predictably $\mathbf{V}$-time-consistent} if,  for any $X\in L^\infty$ and any $t<T$, we have
\begin{equation}\label{time con}
\rho_t(X)=\essinf \{\rho_t(Y.V):\; Y\in L^\infty(\cF_{t+1},\R^{d+1})\text{ and }X-Y\cdot \V\in\cA_{t+1}\}
\end{equation}
\end{defn}

It is easy to check that if $\V\equiv 1$,  then strong time-consistency {\em is} $\V$-time-consistency. In general, strong time-consistency implies $\V$-time-consistency. To see this, first assume strong time-consistency. If we take $Y_0= \rho_{t+1}(X)$ all other components of $Y$ to be zero, so that $\rho_t(Y\cdot \V)=\rho_t(\rho_{t+1}(X))=\rho_t(X)$ and $X-Y\cdot \V=X-\rho_{t+1}(X)\in\cA_{t+1}$. Conversely, $\rho_t(X)=\rho(Z\cdot \V+(X-Z\cdot \V))\leq \rho(Z\cdot \V)+\rho(X-Z\cdot \V)$ so if $X-Z\cdot \V\in \cA_{t+1}$ then $\rho_t(X)\leq \rho_t(Z\cdot \V)$.

We illustrate predictable time-consistency in a finite sample space $\Omega$ with a sign-changed version of Average Value at Risk.
\begin{eg}[Average Value at Risk]\label{eg:AVaR}
Consider the filtered probability space $\Omega=\{1,2,3,4\}$ with $\cF_0$ trivial, $\cF_1= \sigma(\{1,2\}, \{3,4\})$, $\cF_2=2^\Omega=\cF$ (describing a binary branching tree on two time steps).  
Define $\AVaR$, the Average Value at Risk pricing measure, by
\[ \AVaR(X) := \frac{1}{\lambda} \int_0^\lambda q_X(\alpha) \,d \alpha , \]
where $q_X(\alpha)= \inf\{ x \in \R: \P[X \le x] > \alpha \}$.
We may represent $\AVaR$ as
\[  \AVaR(X)= \sup_{\bQ \in \cQ_\lambda}\bE_\bQ[X], \qq{where} \cQ_\lambda = \left\{ \text{probability measures }\bQ \ll \P : \frac{d \bQ}{d \P} \le  \frac{1}{\lambda} \right\},\] 
noting the sign change to make $\AVaR$ a pricing measure; see section 4.4 of \cite{FS11}. We set $\lambda=\frac{1}{50}$, while the objective measure  is given by
\[ \P[\{ 1\}] = \frac{1}{100}, \qquad \P[\{ 2\}] =\P[\{ 3\}] = \frac{9}{100}, \q{and} \P[\{ 4\}] = \frac{81}{100}. \]
For notational convenience, we  represent a probability measure $\bQ$ by the quartuple of its values on atoms, $\bQ(\{i\})=:q_i$, and similarly we write $X(i) = x_i$ for a random variable $X: \Omega \to \R$.  It is easy to see that the representing set $\cQ_\lambda$ is
\[ \cQ_\lambda = \{ \bQ =(q_1,q_2,q_3,q_4): \sum_{i=1}^4 q_i = 1, \quad 0\le q_1\le \half, q_i \in [0,1] \for i=2,3,4 \}. \]
$\cQ_\lambda$ is the convex hull of 6 points:
\begin{align*}
\cQ_\lambda = \conv \{ &(\textstyle\half, \half, 0,0), \qquad  (\half,  0,\half, 0),\qquad (\half,  0,0,\half)  
\\& (0,1,0,0), \qquad (0,0,1,0) , \qquad (0,0,0,1) \}
\end{align*}
 The set of time-0 acceptable claims is
 \[ \cA_0 = \{ X=(x_1,x_2,x_3,x_4) : \sum_{i=1}^4 q_i x_i \le 0 \quad \for \bQ \in \cQ_\lambda\} .\]
Clearly, $X\in \cA_0$ if and only if $\sum_{i=1}^4 q_i x_i \le 0 $ for each of the six extreme points $\bQ$ of $\cQ_\lambda$. These six inequalities are neatly summarised as
\[  \cA_0 =\{ X=(x_1,x_2,x_3,x_4) : x_i \le 0 \for i=1,2,3,4; \q{or} x_1 \ge 0 \text{ and }  x_i\le -x_1 \for i=2,3,4\}.\]
 Define $X^0 := \ind 1 - \ind{2,3,4}$. Then it is clear that
\[\cA_0=\{ \alpha X^0 - \beta: \alpha \ge 0, \, \beta \in L^\infty_+\}.\]

The time-1 acceptance set is 
\begin{align*}
 \cA_1 &= \{ X=(x_1,x_2,x_3,x_4) :  q_1 x_1 +q_2 x_2 \le 0  \q{and} q_3 x_3 +q_4 x_4 \le 0 \quad \for \bQ \in \cQ_\lambda\}
\\&= L^\infty_-.
\end{align*}

\paragraph{Claim}$(\AVaR_0,\AVaR_1)$ is not time-consistent.
\begin{proof}
It is easy to check that we have $\AVaR_0(X^0)=0$, $\AVaR_1(X^0)=\ind{1,2}-\ind{3,4}$,  and thus
\[\textstyle\AVaR_0(\AVaR_1(X^0)) = \AVaR_0(\ind{1,2}-\ind{3,4})=1 >0=\AVaR_0(X^0).\]
\end{proof}

Now we set $\mathbf{V}=(v^0, v^1)$, where  $v^0 \equiv 1$ by convention, and $v^1=X^0+2$, so that
\[ v^1 = 3 \ind 1 +\ind{2,3,4}>0.\]
\paragraph{Claim}$\AVaR$ is predictably $\mathbf{V}$-time-consistent.
\begin{proof}
For any acceptable risk $X\in \cA_0$ we may set $X=\alpha X^0 - \beta$, where $\beta$ is some non-negative random variable taking the value 0 on the event $\{1\} $. We reserve for $X$ by holding $\alpha$ in $v^1$ and $-2\alpha$ in cash $v^0$, giving a mapping $Y_0$ from acceptable risks $X$ to initial reserving portfolios in $\V$: 

\begin{equation}\label{eq:AVaRportf0}
Y_0 = \left(\begin{array}{c} -2\alpha \\ \alpha \end{array}\right).
\end{equation}
Clearly $Y_0\cdot\V = \alpha X^0$.

Set
\begin{equation}\label{eq:AVaRportf1}
 Y_1 = \left(\begin{array}{c} -(2 \alpha+\frac{3}{2}\beta(2))\ind{1,2}-(\alpha + \beta(3)\wedge \beta(4))\ind{3,4} \\ (\alpha + \half \beta(2))\ind{1,2} \\\end{array}\right),
\end{equation}
so that 
\[ Y_1\cdot \mathbf{V} = \alpha X^0-\beta(2)\ind{2} -\beta(3)\wedge \beta(4)\ind{3,4}. \]
Now, we have $\AVaR_0(X-Y_0 \cdot \V) \le 0$, $\AVaR_1(X-Y_1 \cdot \V) \le 0$ and $\AVaR_0(Y_0\cdot \V)=\AVaR_0(X)$ and $\AVaR_1(Y_1\cdot \V)= \rho_1(X)$ as required.
Thus $\AVaR$ is predictably $\mathbf{V}$-time-consistent.

\end{proof}

\end{eg}

\subsection{Predictable representability}
\label{sec:representability}
Given any cone $\cD $ in $\ul L^\infty$ and our vector $\mathbf{V}$ of num\'eraires, we define the collection of portfolios attaining $\cD $ to be
\[ \cD (\mathbf{V}) = \{ Y \in \ul{L}^\infty(\Omega,\cF,\P;\R^{d+1}): Y\cdot \mathbf{V}\in \cD  \}.\]
The set of time-$t$ acceptable portfolios that are $\cF_{t+1}$-measurable is $ K_t(\ul{\cAn}, \mathbf{V}) := \ul{\cAn}(\mathbf{V}) \cap \ul{\cL}^\infty_{t+1}(\R^{d+1})$.

\begin{defn}\label{defn:predrep}The cone $\ul{\cAn}(\mathbf{V})$ is \emph{predictably decomposable} if
\[ \ul{\cAn} (\mathbf{V}) = \ol{\oplus_{t=0}^{T-1} K_t(\ul{\cAn}, \mathbf{V})},\]
where the closure is taken in the weak$^*$-topology. In this case, we say that the cone $\ul{\cAn}$ is \emph{predictably represented} by $\mathbf{V}$.
\end{defn}

\begin{eg}[continues=eg:AVaR][Average Value at Risk] We return to the setting of \Cref{eg:AVaR}. 
\paragraph{Claim}The acceptance set $\cA_0$ is not predictably represented by $1$.
\begin{proof}
 We note that
 \begin{align*}
 K_0(\cA_0,1) &= \{ X \in L^\infty(\cF_1): X \in \cA_0\} = L^\infty_-(\cF_1)
 \\K_1(\cA_0,1)&=\cA_1= L^\infty_-
 \end{align*}

If $\cA_0$ is to be predictably represented  by $1$, we must have that $\cA_0= K_0(\cA_0,1)+K_1(\cA_0,1) =L^\infty_-$; however $\cA_0$ contains $X^0$ which is not in $L^\infty_-$.
\end{proof}
Now set $\mathbf{V}=(1, 3 \ind 1 +\ind{2,3,4})$ as before.
\paragraph{Claim}The set $\cA_0$  is predictably represented by $\mathbf{V}$.
\begin{proof}
 For any $X \in \cA_0$ we may write $X=\alpha X^0 - \beta$, for $\alpha \ge 0$ and $\beta \in L^\infty_+$. Defining $\pi_0 =  Y_0$ and $\pi_1 = Y_1 - Y_0$ for $Y_0$, $Y_1$ as in \cref{eq:AVaRportf0,eq:AVaRportf1}, we have that $X \le \pi_0 \cdot \V + \pi_1 \cdot \V \in K_0(\cA_0, \mathbf{V}) \oplus K_1(\cA_0, \mathbf{V})$. Any non-positive random variable is in any of the $K_t(\cA_0, \mathbf{V})$ for $t=0,1$, so $X$ is in the sum, proving that $\cA_0 \subseteq K_0(\cA_0, \mathbf{V}) \oplus K_1(\cA_0, \mathbf{V})$. The reverse inclusion is clear.
 \end{proof}
\end{eg}

\subsection{Stability properties}

We recall Delbaen's m-stability condition, on a standard stochastic basis $(\Omega,\cF,(\cF_t)_t,\P)$:
\begin{defn}[Delbaen \cite{D06}] A set of probability measures $\cS \subset L^1(\Omega,\cF,\P)$ is \emph{m-stable} if for elements $\bQ^Z \in \cS$ and  $\P \sim\bQ^W \in \cS$, with associated density martingales $Z_t = \condE{\frac{d\bQ^Z}{d\P}}{\cF_t}$ and $W_t = \condE{\frac{d\bQ^W}{d\P}}{\cF_t}$, and for each stopping time $\tau$, the martingale $L$ defined as
\[ L_t = \begin{cases}Z_t &\for t\le \tau \\  \frac{Z_\tau}{W_\tau}W_t &\for t \ge \tau \end{cases}\]
defines an element in $\cS$.
\end{defn}

Note that a set $\cS$ is m-stable if, whenever $\tau$ is a stopping time, and $Z, W \in \cS$ are such that  $Z_\tau = \alpha W_\tau$, then $\alpha W \in \cS$.  Just take $\alpha =  \frac{Z_\tau}{W_\tau}$, and then $L = \alpha W$ in the above definition. We now define a vector-valued generalisation of m-stability, for a subset $\uD  \subset \ul{\cL}^{1}_+(\R^{d+1})$. 

\begin{defn}\label{defn:stable}The subset $\uD  \subset \ul{\cL}^{1}_+(\R^{d+1})$ is \emph{predictably m-stable} if, whenever $\tau \le T$ is a stopping time, and whenever $Z, W \in \uD $ with 
\begin{equation}\label{pred}
\condE{Z}{\ul{\cF}_\tau}=\alpha \condE{W}{\ul{\cF}_\tau} ,
\end{equation}
for some scalar $\alpha$, then $\alpha W$ is also in $\uD $.
\end{defn}
Note that (\ref{pred}) implies that $\alpha$  is $\ul{\cF}_\tau$-measurable and non-negative.
\begin{rmk}\label{time}
If $(\cG_s)_{s=t,\ldots ,S}$ is a filtration with $\cG_u\subset \cF_u$ for each $u$, and $\uD$ is predictably $m$-stable with respect to $(\cF)$ then it is also the case that $\uD$ is predictably $m$-stable with respect to $(\cG)$.
\end{rmk}
\begin{defn}
The cone $\cD \subset \ul L^1_+$ is said to be \emph{predictably $\mathbf{V}$-m-stable} if $\cD  \mathbf{V} = \{Y \mathbf{V}: Y \in \cD  \}$ is predictably m-stable.
\end{defn}

\begin{rmk}\label{rmk:checkstable}
In the case $d=0$, we have $\mathbf{V}\equiv 1$ and so the requirement that a set of Radon-Nikodym derivatives $\cD  \subset \ul L^1_+$ is $1$-m-stable is precisely the requirement that $\cD $ is m-stable. 
\end{rmk}

Every random vector $Z$ in $\cA_0(\mathbf{V})^*$ can be written as a multiple of $\mathbf{V}$, that is, $Z=\wt Z \mathbf{V}$ with $\wt Z \in \cA_0^*$. 

\begin{lem}\label{lem:sackV} Suppose that $\mathbf{V}$ is a collection of $d+1$ num\'eraires, and $\cD $ is a convex cone in $L^\infty$. Then \[ \cD (\mathbf{V})^*=\cD ^*\mathbf{V}.\]
\end{lem}

\begin{proof}
First take $Z\in \cD ^*$. For any $X\in \cD (\mathbf{V})$ we have $\bE[ZV\cdot X] \le 0$ and so $ZV \in \cD (\mathbf{V})^*$, thus $\cD (\mathbf{V})^* \supseteq \cD ^*\mathbf{V}$. 

For the reverse inclusion, recall that $e_i$ denotes the $i$th canonical basis vector in $\R^{d+1}$. First, since $\mathbf{V}\cdot \alpha(v^ie_j-v^je_i)=0$, we have
\[ \alpha(v^ie_j-v^je_i) \in \cD (\mathbf{V}) \qquad \forall \alpha \in L^\infty.\]
Take $Z \in \cD (\mathbf{V})^*$. Now, for any $i,j \in \{1,\dots,d\}$, $\alpha \in L^\infty$, we have
\[ \bE[Z \cdot  \alpha(v^ie_j-v^je_i)] \le 0.\]
Reversing $i$ and $j$ in the above, we may write $\bE[Z \cdot  \alpha(v^ie_j-v^je_i)] = 0$, and allowing first $\alpha= \ind{Z \cdot  (v^ie_j-v^je_i) >0}$ then $\alpha= \ind{Z \cdot  (v^ie_j-v^je_i) <0}$, we see that in fact,
\[ Z\cdot  (v^ie_j-v^je_i) = 0 \qquad \text{a.s. for any }i,j,\]
and so, taking $i=0$ we have $ Z^j=Z^0v^j $ a.s. for each $j$, thus any $Z \in \cD (\mathbf{V})^*$ must be of the form $Z^0\mathbf{V}$ for some $Z^0 \in L^1$. Now, given $C\in \cD $, take $X$ such that $X\cdot \mathbf{V}=C$ (which implies that $X\in \cD (\mathbf{V})$), then 
\[0 \ge \bE[W\mathbf{V}\cdot X]=\bE[WC],\]
and since $C$ is arbitrary, it follows that $W\in \cD ^*$. Hence $\cD (\mathbf{V})^*\subseteq \cD ^*\mathbf{V}$.
\end{proof}

\begin{rmk}\label{rmk:predstabchecker}
In light of \Cref{lem:sackV}, we may check that $\cA_0(\mathbf{V})^* \equiv \cA_0^* \mathbf{V}$ is predictably stable in the following way. We first associate to each $Z\in \cA_0^*$ the probability measure $\bQ^Z$, defined through its Radon-Nikodym derivative \[\frac{d \bQ^Z}{d \P} = \frac{Z}{\bE[Z]}.\] 
We note that if ${Z},{W} \in \cA_0(\mathbf{V})^*$, then we may find $ \wt Z, \wt W \in \ul\cA_0^*$ such that ${Z}=\wt Z \mathbf{V}$ and ${W}=\wt W \mathbf{V}$. The assumption that $v^0\equiv 1$ gives the equivalence of the condition $\condE{{Z}}{\ul{\cF}_\tau}=m \condE{{W}}{\ul{\cF}_\tau}$ with the condition
\begin{equation}
\bE_{\bQ^{\wt Z}}[\mathbf{V}|\ul{\cF}_\tau] =\bE_{\bQ^{\wt W}}[\mathbf{V}|\ul{\cF}_\tau].  \label{eq:wtreduction}
\end{equation}
The set $\cA_0(\mathbf{V})^* $ is predictably $\mathbf{V}$-m-stable if, for any stopping time $\tau \le T$, whenever $\wt Z, \wt W \in \cA_0^*$ are such that \eqref{eq:wtreduction} holds, then \[ \frac{\condE{\wt{Z}}{\ul{\cF}_\tau}}{\condE{\wt{W}}{\ul{\cF}_\tau}}{W}\in \cA_0^*(\mathbf{V}).\] 

\end{rmk}

\begin{eg}[continues=eg:AVaR] \label{eg:stabAVaR}We return to the setting of \Cref{eg:AVaR}.

\paragraph{Claim}$\cA_0^*$ is not m-stable.
\begin{proof}
 Define measures $\bQ^1=(\half, \half, 0,0) \in \cQ_\lambda$ and $\bQ^2= (\half,  0,\half, 0) \in \cQ_\lambda$. We form the time-1 pasting of the measures  $\bQ^1$ and $\bQ^2$ by setting
\[ \frac{d\wt\bQ}{d\P}=\frac{\condE{\frac{d\bQ^1}{d \P}}{\cF_1}}{\condE{\frac{d\bQ^2}{d \P}}{\cF_1}} \frac{d\bQ^2}{d \P}  \]
so that $\wt \bQ = (1,0,0,0)$.
Now $ \wt q_1=1>\half$ which shows $\wt \bQ \not \in \cQ_\lambda$, and so $\cQ_\lambda$ is not m-stable.
\end{proof}

Now set $\mathbf{V}=(1, 3 \ind 1 +\ind{2,3,4})$ as before.

\paragraph{Claim}$\cA_0^*$ is $\mathbf{V}$-m-stable.
\begin{proof}
First, consider the pasting $ \wt \bQ = \bQ \oplus_\tau \bQ'$ of measures $\bQ$ and $\bQ'$ in $\cQ_\lambda$ at the  stopping time $\tau$:
\begin{align*}
\frac{d\wt\bQ}{d\P} &= \frac{\condE{\frac{d\bQ}{d \P}}{\cF_\tau}}{\condE{\frac{d\bQ'}{d \P}}{\cF_\tau}} \frac{d\bQ'}{d \P}
\\& = \frac{d\bQ'}{d \P} \ind{\tau = 0} + \frac{\condE{\frac{d\bQ}{d \P}}{\cF_1}}{\condE{\frac{d\bQ'}{d \P}}{\cF_1}} \frac{d\bQ'}{d \P}\ind{\tau = 1} + \frac{d\bQ}{d \P} \ind{\tau = 2}.
\end{align*}
 By  \Cref{rmk:predstabchecker}, we fix $\wt Z$ and $\wt Z'$ in $\cA_0^*$ with associated probability measures $\bQ$ and $\bQ'$ that additionally satisfy 
 \[ \bE_{\bQ}[v^1|\ul{\cF}_\tau] = \bE_{\bQ'}[v^1|\ul{\cF}_\tau],\]
 and we aim to show that $\wt \bQ \in \cQ_\lambda$. On the event $\{\tau = 0\}$ (respectively $\{\tau = 2\}$), we have that $\wt\bQ = \bQ'$ (respectively   $\wt\bQ = \bQ$) and the bound $\wt \bQ(1)\le \half$ is trivially satisfied.
The event $\{\tau = 1\}$ is one of $\emptyset$, $\{1,2\}$, $\{3,4\}$, $\Omega$.
Writing $\bQ = (q_i)_{i=1}^4$, for $\omega \in \{1,2,3,4\}$,
\[ \bE_{\bQ}[v^1|\ul{\cF}_1](\omega) = \frac{3q_1+q_2}{q_1+q_2}\ind{q_1+q_2>0}\ind{1,2}(\omega)+\ind{q_3+q_4>0}\ind{3,4}(\omega) \]
We may paste measures $\bQ$ and $\bQ'$ that satisfy
\begin{equation*}
\frac{3q_1+q_2}{q_1+q_2}\ind{q_1+q_2>0}\ind{1,2} +\ind{q_3+q_4>0}\ind{3,4}= \frac{3q_1'+q_2'}{q_1'+q_2'}\ind{q_1'+q_2'>0}\ind{1,2} +\ind{q_3'+q_4'>0}\ind{3,4} 
\end{equation*}
on $\{\tau = 1\}$, which simplifies  to the requirement that
 \begin{equation}
\label{eq:simpcond}
\frac{q_1}{q_2} \ind{q_1+q_2>0} \one{\{1,2\}\cap\{\tau = 1\}} +\ind{q_3+q_4>0}\one{\{3,4\}\cap\{\tau = 1\}} = \frac{q_1'}{q_2'} \ind{q_1'+q_2'>0}\one{\{1,2\}\cap\{\tau = 1\}} +\ind{q_3'+q_4'>0}\one{\{3,4\}\cap\{\tau = 1\}} .
\end{equation}
On $\{\tau = 1\}\supset \{ 1\}$, the pasting $\wt\bQ$ weights $\{1\}$ as 
\[ \bQ\oplus_{\tau} \bQ' (\{1\})= (q_1+q_2) \frac{q_1'}{q_1'+q_2'}\ind{q_1'+q_2'>0} = (q_1+q_2) \frac{\frac{q_1'}{q_2'} }{\frac{q_1'}{q_2'}+1}\ind{q_1'+q_2'>0} \stackrel{\text{\eqref{eq:simpcond}}}{=} q_1\ind{q_1+q_2>0}\]
The other cases are easy to check. Thus $\bQ\oplus_{\tau} \bQ' \in \cA_0^*$, and  $\cA_0^*$ is $\mathbf{V}$-m-stable.
\end{proof}

\end{eg}

\subsection{Main result}
We fix   num\'eraires $\mathbf{V}$,  a coherent risk measure  $\rho=(\rho_t)_t$ with convex representing set of probability measures $\cQ$, and take $\ul{\cA}_t$ to be the acceptance set of $\rho_t$ for $t \in \bT$. The main result is

\begin{theorem}\label{thm:main}
The following are equivalent:
\begin{enumerate}[(i)]
\item
$(\ul\rho_t)_{t\in \bT}$ is predictably $\mathbf{V}$-time-consistent;
\item
$\ul{\cAn}$ is predictably represented by  $\mathbf{V}$;
\item
 $\ul{\cAn}(\mathbf{V})^*  $ is predictably m-stable.
\end{enumerate}
\end{theorem}

The proof will be given in \Cref{sec:pomr}.

We now highlight waypoints in the proof of \Cref{thm:main}.

Thinking of the conditional expectation $\bE[\cdot|\ul\cF_{t+1}]$ as a projection from $\ul{\cL}^1(\R^{d+1})$ to  $\ul{\cL}^1_{t+1}(\R^{d+1})$, we define the \emph{predictable pre-image} of $\uD$ at time $t$ by first projecting $\uD$ to $\ul{\cL}^1_{t+1}(\R^{d+1})$, then taking the $\cF_t$-cone, and finally taking the pre-image under the projection $\bE[\cdot|\ul\cF_{t+1}]$. The $\cF_t$-cone of a set $E$ is 
\[\cone_{\cF_t}E = \{ \alpha w_1 + \beta w_2 : \alpha,\beta \in L^\infty_+(\cF_t), \, w_1,w_2 \in E\}.\]
More concisely:
\begin{defn}
For $\uD \subset \ul{\cL}^1_+(\R^{d+1})$, we define for each time $t$ the \emph{predictable pre-image} of $\cD$ by
\begin{align}\label{eq:MtD }
\cM_t(\uD ) := \{ Z \in \ul{\cL}^1(\R^{d+1}):  &\exists \alpha_t \in \ul{L}^0_{t,+},  \exists Z'\in \uD  \nonumber
\\&\text{ such that } \alpha_t Z' \in \ul{\cL}^1(\R^{d+1}) \text{ and } \condE{Z}{\ul{\cF}_{t+1}}= \alpha_t \condE{Z'}{\ul{\cF}_{t+1}} \}.
\end{align}
\end{defn}

The predictable pre-image of  a set $\uD \subset \ul{\cL}^1_+(\R^{d+1})$, is key to understanding predictably stable convex cones, as shown in the following three lemmas.
The first gives  an alternative characterisation of stability:
\begin{lem}\label{lem:eqstab} Let $\uD  \subset \ul{\cL}^{1}_+(\R^{d+1})$. The following are equivalent:
\begin{enumerate}[(i)]
\item
for each $t\in \{0,1,\dots, T\}$, whenever $Y,W \in \uD $ are such that there exists $Z\in \uD $, a set $F\in \ul{\cF}_t$, positive random variables $\alpha, \beta \in \ul{\cL}^0(\ul{\cF}_t)$ with $\alpha Y, \beta W \in \ul{\cL}^1(\R^{d+1})$ and 
\[ X:= \one{F} \alpha Y + \one{F^c} \beta W \q{satisfies}\condE{X}{\ul{\cF}_t}=\condE{Z}{\ul{\cF}_t},\]
 then $X$  is also a member of $\uD $;
\item
$\uD $ is predictably stable, that is, for each stopping time $\tau \le T$, whenever $Z,W\in \uD $ are such that 
\[ \condE{Z}{\ul{\cF}_\tau}=m \condE{W}{\ul{\cF}_\tau} ,\]
then $mW$ is also a member of $\uD $.
\end{enumerate}
\end{lem}

\begin{proof}
\emph{(ii)$\implies$(i):} We suppose that (ii) holds, and fix $t \in \bT$. We aim for a triple of random variables $Y,W,Z$ in $\uD $, together with an $F \in \ul{\cF}_t$, and $\alpha, \beta$ as required in condition (i), such that we can apply (ii) twice to show that the resulting $X$ defined in condition (i) is a member of $\uD $.

First, let $\tau=T \one{F} + t \one{F^c}$ and suppose $Z, W\in \uD $ satisfy $ \condE{Z^i}{\ul{\cF}_\tau}=m \condE{W^i}{\ul{\cF}_\tau}$ for all $i$. By (ii), we have $\wt{X}:= m W \in \uD $. Writing 
\[ \beta:=\frac{\condE{Z^i}{\ul{\cF}_t}}{\condE{W^i}{\ul{\cF}_t}}\ind{\condE{W^i}{\ul{\cF}_t}> 0}, \]
we may express $\wt{X} = Z \one{F}+ \beta W\one{F^c}$.

Second, let $\wt \tau=t \one{F} + T \one{F^c}$ and suppose $Y\in \uD $ satisfies $ \condE{\wt X^i}{\ul{\cF}_\tau}=\wt m \condE{Y^i}{\ul{\cF}_\tau}$ for all $i$. By (ii), we have ${X}:= \wt m Y \in \uD $. Writing 
\[ \alpha:=\frac{\condE{\wt X^i}{\ul{\cF}_t}}{\condE{Y^i}{\ul{\cF}_t}}\ind{\condE{Y^i}{\ul{\cF}_t}> 0},\]
we may express ${X} = \alpha Y \one{F}+ \beta W\one{F^c}$.

Now, we have a $t$ fixed, $Y,W,Z \in \uD $, a set $F\in \ul{\cF}_t$, and positive r.v.s $\alpha, \beta \in \ul{\cL}^0(\ul{\cF}_t)$. We have already that $X\in \uD $, thus it remains to check\footnote{the integrability conditions $\alpha Y, \beta W \in \ul{\cL}^1(\R^{d+1})$ are easily verified.} that $X$ and $Z$  as defined above satisfy $\condE{X}{\ul{\cF}_t}=\condE{Z}{\ul{\cF}_t}$. 
\begin{align*}
\condE{X}{\ul{\cF}_t} &= \one{F} \condE{\alpha Y}{\ul{\cF}_t} +\one{F^c} \condE{\beta W}{\ul{\cF}_t}
\\&=\one{F} \condE{\frac{\condE{\wt X^i}{\ul{\cF}_t}}{\condE{Y^i}{\ul{\cF}_t}} \ind{\condE{Y^i}{\ul{\cF}_t}> 0}Y}{\ul{\cF}_t} 
\\&\qquad+\one{F^c} \condE{\frac{\condE{Z^i}{\ul{\cF}_t}}{\condE{W^i}{\ul{\cF}_t}} \ind{\condE{W^i}{\ul{\cF}_t}> 0}W}{\ul{\cF}_t}
\\&=\one{F} \condE{\wt X}{\ul{\cF}_t} +\one{F^c} \condE{Z}{\ul{\cF}_t}
\\&= \condE{Z}{\ul{\cF}_t},
\end{align*}which establishes statement (i).

\emph{(i)$\implies$(ii):} Say (i) holds; then (ii) holds  for when $\tau = T$ trivially. Now suppose that (ii) holds for any stopping time $\tau \ge k+1$ a.s., and proceed by backward induction on the lower bound of the stopping times. Fix an arbitrary stopping time $\wt \tau \ge k$ a.s., and define $F=\{ \wt \tau \ge k+1\}$ and the stopping time $\tau^* := \wt \tau \one F + T \one{F^c}$. Note that $\tau^* \ge k+1$, since $F^c = \{ \wt \tau = k\}$.

We shall now take $Z, W \in \uD $ that satisfy $\condE{Z^i}{\ul{\cF}_{\wt\tau}}=m \condE{W^i}{\ul{\cF}_{\wt\tau}}$ for all $i$, and aim to show that $mW$ is indeed an element of $\uD $, with the help of condition (i).

To this end, define 
\[ Y:= W \frac{\condE{Z^i}{\ul{\cF}_{\tau^*}}}{\condE{W^i}{\ul{\cF}_{\tau^*}}}\ind{\condE{W^i}{\ul{\cF}_{\tau^*}}> 0} =\one F W  \frac{\condE{Z^i}{\ul{\cF}_{\wt\tau}}}{\condE{W^i}{\ul{\cF}_{\wt\tau}}}\ind{\condE{W^i}{\ul{\cF}_{\wt\tau}}> 0} + Z \one{F^c}. \]
By the inductive hypothesis, $Y$ is in $\uD $, thanks
 to the bound $\tau^* \ge k+1$, . 

Now, we have $t=k$ fixed, $Y,W,Z \in \uD $, a set $F\in \ul{\cF}_t$, and positive random variables $\alpha\equiv 1$, $\beta :=\one{F^c} \frac{\condE{Z^i}{\ul{\cF}_{k}}}{\condE{W^i}{\ul{\cF}_{k}}}$. Define 
\begin{align*}
X&:=\one{F} \alpha Y + \one{F^c} \beta W
\\&=W\one{F}  \frac{\condE{Z^i}{\ul{\cF}_{\wt\tau}}}{\condE{W^i}{\ul{\cF}_{\wt\tau}}}\ind{\condE{W^i}{\ul{\cF}_{\wt\tau}}> 0}  + W\one{F^c} \frac{\condE{Z^i}{\ul{\cF}_{k}}}{\condE{W^i}{\ul{\cF}_{k}}} \ind{\condE{W^i}{\ul{\cF}_{k}}> 0} 
\\&=W  \frac{\condE{Z^i}{\ul{\cF}_{\wt\tau}}}{\condE{W^i}{\ul{\cF}_{\wt\tau}}}\ind{\condE{W^i}{\ul{\cF}_{\wt\tau}}> 0} .
\end{align*}
It is elementary to check that $X$ and $Z$  as defined above satisfy $\condE{X}{\ul{\cF}_k}=\condE{Z}{\ul{\cF}_k}$. Thus by (i), $X$ is an element of $\uD $, which completes the inductive step.
\end{proof}

\begin{lem}\label{lem:A1} Suppose $\uD \subset \ul{\cL}^1_+(\R^{d+1})$. If $\uD $ is a predictably stable convex cone, then
\[ \uD = \bigcap_{t=0}^{T-1} \cM_t(\uD ).\]
\end{lem}

\begin{proof}
The inclusion $\uD  \subset \cap_{t=0}^{T-1} \cM_t(\uD )$ is trivial. In the following, we write $Z|_t$ for $\condE{Z}{\cF_t}$.

Now $Z \in \cap_{t=0}^{T-1} \cM_t(\uD )$, and we aim to show that $Z\in \uD $. So, for all $t \in \{0,1,\dots, T-1\}$, there exist $\beta_t \in  \ul{L}^0_{t,+}$ and $Z^t \in \uD $ such that $\beta_t Z \in \ul{\cL}^1_+(\R^{d+1}) \text{ and } Z|_{t+1}= \beta_t Z^t|_{t+1}$.

Define
\begin{align*}
\xi^{T-1}&=Z^{T-1} \\
\xi^t &= \one{F_t} \kappa_t \xi^{t+1}+\one{F^c_t}Z^t &\for t \in \{0,1,\dots, T-2\},
\end{align*}
where $F_t = \{ \beta_t>0\} $ and $\kappa_t = \beta_{t+1}/\beta_t$.

Note $ Z=Z|_T=\beta_{T-1}Z^{T-1}|_T=\beta_{T-1}\xi^{T-1}$ and 
\[Z=\beta_0 \kappa_0 \kappa_1 \cdots \kappa_{T-2} \xi^{T-1} = \beta_0 \xi^0. \]
Thus we only need to show $\xi^0$ is in the cone $\uD $ to deduce that $Z=\beta_0 \xi^0$ is in $\uD $. 
\paragraph{Claim} For all $t \in \{0,1,\dots, T-1\}$, we have $\xi^t|_{t+1}=Z^t|_{t+1}$ and $\xi^t \in \uD $.

We shall proceed by backwards induction, starting from the observation $\xi^{T-1} = Z^{T-1} \in \uD $. Suppose that for $s \ge t+1$, we have $\xi^s|_{s+1}=Z^s|_{s+1}$ and $\xi^s \in \uD $.
\begin{align*}
\xi^t|_{t+1} &= \condE{\one{F_t} \kappa_t \xi^{t+1}+\one{F^c_t}Z^t}{\ul{\cF}_{t+1}}
\\&=\condE{\one{F_t} \kappa_t Z^{t+1}+\one{F^c_t}Z^t}{\ul{\cF}_{t+1}}
\end{align*}
Now, whilst $\beta_t>0$, i.e. on the event $F_t$,
\[ \condE{ \kappa_t Z^{t+1}}{\ul{\cF}_{t+1}} = \frac{1}{\beta_t}\condE{ \beta_{t+1} Z^{t+1}}{\ul{\cF}_{t+1}} = \frac{1}{\beta_t}\condE{ Z|_{t+2}}{\ul{\cF}_{t+1}} = \frac{Z|_{t+1}}{\beta_t}  = Z^{t}|_{t+1} \]
allowing us to conclude  
\[ \xi^t|_{t+1} = \condE{\one{F_t} Z^{t}|_{t+1}+\one{F^c_t}Z^t}{\ul{\cF}_{t+1}}= Z^{t}|_{t+1}.\]
By hypothesis $\uD $ is stable, so by \Cref{lem:eqstab} we see that $\xi^t \in \uD $.
\end{proof}


\begin{lem}\label{lem:A2} For $\uD \subset \ul{\cL}^1_+(\R^{d+1})$, define
\[ [\uD ] := \bigcap_{t=0}^{T-1}\left( \ol{\conv} \cM_t(\uD )\right), \] 
where $ \cM_t(\uD )$ is as defined in \eqref{eq:MtD }, the symbol $ \ol{\conv}$ denoting the closure in $\ul{\cL}^1_+(\R^{d+1})$ of the convex hull.
\begin{enumerate}[(a)]
\item
$[\uD ]$ is the smallest predictably $m$-stable closed convex cone in $\ul{\cL}^1_+(\R^{d+1})$ containing $\uD $;
\item
$\uD =[\uD ]$ if and only if $\uD $ is a predictably $m$-stable closed convex cone in $\ul{\cL}^1_+(\R^{d+1})$.
\end{enumerate}
\end{lem}

\begin{proof}
It is clear that $[\uD ]$ is a closed convex cone in $\ul{\cL}^1$. To see that $[\uD ]$ is stable, we use the definition of stability according to \Cref{lem:eqstab}. Fix $t\in \{0,1,\dots, T\}$, and suppose $Y,W \in [\uD ]$ are such that there exists $Z\in[\uD ]$, a set $F\in \ul{\cF}_t$, positive processes $\alpha, \beta \in \ul{\cL}^0(\ul{\cF}_t)$ with $\alpha Y, \beta W \in \ul{L}^1(\R^{d+1})$ and 
\[ X:= \alpha Y  \one{F} +  \beta W\one{F^c}\]
satisfies $\condE{X}{\ul{\cF}_t}=\condE{Z}{\ul{\cF}_t}$. We aim to show $X$  is also a member of $[\uD ]$, that is,
\[X \in   \ol{\conv} \cM_s(\uD ) \qquad \forall  0 \le s  \le T-1 .\]

First consider $s \in \{0,1,\dots,  t-1 \}$. From the definition of $\cM_s(\uD )$,
\[ Z \in \conv\cM_s(\uD ) \q{and} \condE{X}{\ul{\cF}_t} = \condE{Z}{\ul{\cF}_t}  \qiq X \in\conv \cM_s(\uD ), \]
since the membership of an integrable $Z$ in $\cM_s(\uD )$ only depends on its conditional expectation $ \condE{Z}{\ul{\cF}_{s+1}}$.
More generally, we show
\[ Z \in  \ol{\conv}\cM_s(\uD ) \q{and} \condE{X}{\ul{\cF}_t} = \condE{Z}{\ul{\cF}_t}  \qiq X \in  \ol{\conv}\cM_s(\uD ). \]
 Take a sequence $(Z^n) \subset \conv \cM_s(\uD )$ such that $Z^n \to Z$ in $\ul{\cL}^1$. Define the sequence
\[ X^n :=  \condE{Z^n}{\ul{\cF}_t} + X-  \condE{X}{\ul{\cF}_t} .\]
Note that $X^n \to X$ as $n \to \infty$ and for each $n$, $\condE{X^n}{\ul{\cF}_t} = \condE{Z^n}{\ul{\cF}_t}$. So $X^n \in \conv \cM_s(\uD )$, thus $X \in  \ol{\conv}\cM_s(\uD )$.

Now consider $s \in \{t,t+1,\dots,  T-1 \}$. We begin by choosing sequences $(Y^n), (W^n) \subset \conv \cM_s(\uD )$ such that $Y^n \to Y$ and $W^n \to W$ in $\ul{\cL}^1$. Define, for $n,K \in \bN$,
\[ X^{n,K}:= \ind{\alpha \le K}  \alpha Y^n\one{F} + \ind{\beta \le K}  \beta W^n\one{F^c}. \]
The fact that $X^{n,K} \in \conv \cM_s(\uD )$ follows from the following two elementary properties:
\begin{enumerate}
\item
if $Z \in \conv\cM_s(\uD )$ and $g \in \ul{\cL}^\infty_+(\ul{\cF}_t)$, then $gZ \in \conv\cM_s(\uD )$;\footnote{Let $Z \in \cM_s(\uD )$. Then 
\begin{align*}
\exists \alpha_t \in \ul{L}^0_{t,+},  \exists Z'\in \uD  &\text{ such that } \alpha_t Z \in \ul{\cL}^1 \text{ and } Z|_{t+1}= \alpha_t Z'|_{t+1}
\\\implies\exists \alpha_tg \in \ul{L}^0_{t,+},  \exists Z'\in \uD  &\text{ such that } \alpha_t gZ \in \ul{\cL}^1 \text{ and } gZ|_{t+1}= \alpha_tg Z'|_{t+1}
\end{align*}and then take convex hulls.} and
\item
if $Z^i \in \conv\cM_s(\uD )$ for $i=1,2$, then $Z^1+Z^2 \in \conv\cM_s(\uD )$.
\end{enumerate}

Now, for any $K$ fixed, $\ind{\alpha \le K}  \alpha Y^n \to \ind{\alpha \le K}  \alpha Y$ as $n \to \infty$, and similarly $ \ind{\beta \le K}  \beta W^n \to  \ind{\beta \le K}  \beta W$. Since $\alpha Y$ and $\beta W$ are integrable, we now send $K \to \infty$ to see that
\[ X = \lim_{K\to \infty} \lim_{n \to \infty} X^{n,K} \in  \ol{\conv} \cM_s(\uD ) \]
which completes the proof that $X$ is indeed a member of $[\cD ]$.

To show minimality of $[\uD ]$ in the class of stable closed convex cones containing $\uD $, we note that if $\uD  \subset \uD '$ then $ [\uD ] \subset [\uD ']$. Taking $\uD '$ to be another stable closed convex cone containing $\uD $, we have $\uD '=[\uD ']$ by \Cref{lem:A1}, and so $\uD '$ contains $[\uD ]$. 
To show the equivalence in statement (b), the forward implication is due to the stability of $[\uD ]$, and the reverse is \Cref{lem:A1}.
\end{proof}

The proof of equivalence of statements (ii) and (iii) of \Cref{thm:main} is underpinned by the following
\begin{theorem}\label{thm:crucialClaim}
For any $t\in \{0,1,\dots,T-1\}$, 
\begin{equation}\label{eq:MandK}
K_t(\ul\cA,\mathbf{V})=(\cM_t(\ul\cA(\mathbf{V})^*))^*.
\end{equation}•
\end{theorem}
We defer the proof until section \ref{sec:pomr}.

Thus we have characterised each ``summand'' in the representation (cf. \cref{defn:predrep}) as a dual set of the predictable pre-image of the dual of the set of acceptable portfolios in $\mathbf{V}$.

\section{Examples}

In this section we present a brief exposition of the versatility of the framework.

\subsection{Modelling transaction costs}
We now present an example motivated by buying and selling a  stock in a market with transaction costs across two time periods ($T=2$).  Let $N_1$ and $N_2$ be two independent and identically distributed standard Gaussian random variables under objective measure $\P$. Fix $M>0$ and define the truncated random variables $\wt N_i := N_i \wedge M$, for $i=1,2$. Define the constant $a_M$ such that $\bE_\P[\exp(\wt N_i - a_M)]=1$:
\[ a_M := \log \bE_\P[\exp(\wt N_1)] = \log\left(e^\half \Phi(M-1) + e^M (1-\Phi(M))\right).\]
Define the filtration by $\cF_0$ trivial, $\cF_1 = \sigma(\wt N_1)$, and $\cF_2 = \sigma(\wt N_1,\wt N_2)$. 

The market consists of a ``cash account'' $v_0 \equiv 1$ and a ``stock'' with time-2 price \[v_1 = \exp\left(\wt N_1 + \wt N_2 - 2a_M\right).\]

 Set $\mathbf{V} = (v_0,v_1)$. At time 0, to buy 1 unit of $v_1$ a purchaser must pay $1+\lambda$ cash, and to sell 1 unit of $v_1$ a vendor receives $1-\lambda$. At time 1, knowing the value of $\wt N_1$, buying 1 unit of $v_1$  costs $(1+\lambda)e^{\wt N_1-a_M}$, and selling 1 unit of $v_1$ makes $(1-\lambda)e^{\wt N_1-a_M}$. Define the $\cF_1$-cone of a set $E$ by $\cone_{\cF_1}E = \{ \alpha w : \alpha \in L^\infty_+(\cF_1), \, w \in \conv E\}$.   If we also allow wealth to be consumed, we arrive at the following set of claims to which we may trade from zero initial wealth:
\begin{align*}
 \cA =&\cone \left\{( -(1+\lambda) , 1 ) ,  ( 1-\lambda , -1 )\right\} \cdot \mathbf{V}
\\&\oplus \cone_{\cF_1}\left\{( -(1+\lambda)e^{\wt N_1-a_M} , 1 ) ,  ( (1-\lambda)e^{\wt N_1-a_M} , -1 )\right\} \cdot \mathbf{V} 
\\&\oplus(- L^\infty_+).
\end{align*}
In the sum, the first term describes those claims that can be realised at time 0, the second term describes those claims that can be realised at time 1, and the last describes consumption of wealth at any time.
It is easy to show that the dual of $ \cA$ is
\[ \cQ := \left\{ \bQ \ll \P : \bE_\bQ[v_1] \in [1-\lambda,1+\lambda] \q{and} \bE_\bQ[\exp(\wt N_2 - a_M)|\cF_1] \in [1-\lambda,1+\lambda] \right\}.\]
Note that $\cQ$ is a convex
set of probability measures that is not m-stable. Define a coherent risk measure  by $\rho_t(X) = \sup_{\bQ \in \cQ}\bE_\bQ[X|\cF_t]$ for $t = 0,1$.  We have $\rho_0(v_1) = 1+\lambda$, but
\[  \rho_1(v_1) = e^{\wt N_1-a_M}\sup_{\bQ \in \cQ} \bE_\bQ[e^{\wt N_2 - a_M}|\cF_1] = (1+\lambda)e^{\wt N_1-a_M},\] 
and so \[\rho_0(\rho_1(v_1)) = (1+\lambda)\sup_{\bQ \in \cQ}\bE_\bQ[e^{\wt N_1-a_M}]  = \frac{(1+\lambda)^2}{1-\lambda} > \rho_0(v_1). \]
The last line follows from the inequalities for any $\bQ \in \cQ$:
\[ 1+\lambda \ge \bE_\bQ[v_1] = \bE_\bQ[e^{\wt N_1-a_M} \bE_\bQ[e^{\wt N_2-a_M}|\cF_1]] \ge (1-\lambda)\bE_\bQ[e^{\wt N_1-a_M}].\]

Now, we may show that $\cQ$ must be $\mathbf{V}$-m-stable: we take  two measures $\bQ^\Lambda$ and $\bQ^M$ with Radon-Nikodym derivatives $\Lambda$ and $M$, form the pasting at a stopping time $\tau\in \{0,1,2\}$, and check that the pasted measure $\wt \bQ$, defined by
\[ \frac{d \wt\bQ}{d\P} = \frac{M}{\bE[M|\cF_\tau]} \bE[\Lambda |\cF_\tau] \]
  is also in $\cQ$.  
  Noting that $\ind{\tau=2} = 1 - \ind{\tau \le 1} \in L^\infty_1$, we calculate
\begin{align*}
\bE_{\wt\bQ}[\exp(\wt N_2 - a_M)|\cF_1] &= \bE\left[ \left.\left( \frac{M}{\bE[M|\cF_1]}\ind{\tau \le 1} + \frac{\Lambda}{\bE[\Lambda|\cF_1]}\ind{\tau=2} \right)\exp(\wt N_2 - a_M)\right|\cF_1 \right]
\\&= \bE_{\bQ^M}[\exp(\wt N_2 - a_M)|\cF_1]\ind{\tau \le 1} +\bE_{\bQ^\Lambda}[\exp(\wt N_2 - a_M)|\cF_1]\ind{\tau=2},
\end{align*}•
  so we see that the condition $\bE_{\wt\bQ}[\exp(\wt N_2 - a_M)|\cF_1] \in [1-\lambda,1+\lambda]$ is  satisfied. To satisfy the definition of $\mathbf{V}$-m-stability, we need only check that $\wt \bQ \in \cQ$ for those $\bQ^\Lambda$ and $\bQ^M$ that satisfy
  \[ \bE_{\bQ^\Lambda}[v_1 | \cF_\tau] = \bE_{\bQ^M}[v_1 | \cF_\tau]. \]
Hence, we now calculate
\begin{align*}
\bE_{\wt\bQ}[v_1] &= \bE\left[ \frac{\bE[\Lambda |\cF_\tau]}{\bE[M|\cF_\tau]} \bE[Mv_1 |\cF_\tau] \right]
\\&=\bE\left[ {\bE[\Lambda |\cF_\tau]} \bE_{\bQ^M}[v_1 |\cF_\tau] \right]
\\&=\bE\left[ {\bE[\Lambda |\cF_\tau]} \bE_{\bQ^\Lambda}[v_1 |\cF_\tau] \right]
\\&=\bE_{\bQ^\Lambda}[v_1].
\end{align*}•
Thus $\cQ$ is $\mathbf{V}$-m-stable.

\subsection{A Haezendonck--Goovaerts risk measure}
The following is an example employing the so-called Haezendonck--Goovaerts risk measures; we refer the reader to the work of Bellini and Rosazza Gianin \cite{BRG08}. Consider a two-period binary branching tree, with $\P\{\omega\}=\frac{1}{4}$ for all four elements $\omega \in \Omega$.  We choose a (normalised) Young function $\Phi(x)=x^2$, and define the Orlicz premium principle to be the unique solution $H_\alpha(X)$ of the equation
\[ \bE\left[ \Phi \left( \frac{X}{H_\alpha(X)} \right)\right] = 1 - \alpha \qquad \for X \neq 0; \qquad H_\alpha(0):=0. \]
Fix $\alpha = \half$, and rearrange the above to see $H_\half(X)=\sqrt{2} \|X\|_2 = \sqrt{2} \left(\bE\left[ X^2 \right]\right)^\half$.
We now define the \emph{Haezendonck measure}  to be
\[ \rho_0(X)=\sup_{\bQ \in \cQ} \bE_\bQ[X] \qq{where} \cQ:=\{ \bQ \ll \P: \bE_\bQ[Y] \le H_\half(Y) \quad \forall Y \in L^\infty_+ \} .\]
We write $\bQ(\{i\})=: q_i$ for a measure $\bQ$ on $(\Omega,\cF)$, and $X(i)=x_i$ for a random variable $X$ on $(\Omega,\cF,\P)$. First, we characterise $\cQ$. Note that the constraint in the definition of $\cQ$ implies
\[ \sup_{0\neq Y \in L^\infty_+ }\bE\left[ \frac{d\bQ}{d\P} \frac{Y}{\|Y\|_2}\right] \le \sqrt{2}. \]
The supremum is attained upon choosing $Y=\frac{d\bQ}{d\P}$, so the above inequality implies $\left\|\frac{d\bQ}{d\P}\right\|_2^2 \le 2$, 
thus
\[ \cQ=\left\{ \bQ=(q_1,\dots, q_4): q_i \ge 0, \quad \sum_{i=1}^4 q_i=1,\quad \sum_{i=1}^4 q_i^2 \le \half\right\}.\]
\paragraph{$\cQ$ is not m-stable} Define measures $\bQ^\Lambda$ and $\bQ^M$ from $\Lambda_2=2 \times \ind{1,2}$ and $M_2=2 \times \ind{1,3}$ respectively. We see that both are elements of $\cQ$, and their restrictions to $(\Omega, \cF_1)$ are described by $\Lambda_1=\bE[\Lambda_2|\cF_1] = 2 \times \ind{1,2}$, and  $M_1=\bE[M_2|\cF_1] =  1$. We form the time-1 pasting of the measures  $\bQ^\Lambda$ and $\bQ^M$ by setting
\[ \frac{d\wt\bQ}{d\P}=\frac{\Lambda_1}{M_1} M_2=4 \times \ind{1}.\]
Here, $ \sum_{i=1}^4\wt q_i^2=1>\half$, so $\wt \bQ \not \in \cQ$, and the set $\cQ$ is not m-stable.

Now, set 
\[ \mathbf{V} = \left( 1, \sqrt{2} \ind{1} + 1, \sqrt{2} \ind{3} + 1  \right).\]
\paragraph{$\cQ$ is $\mathbf{V}$-m-stable}
We calculate $\bE_{\bQ}[\mathbf{V}|\ul{\cF}_1]$ as in \Cref{eg:stabAVaR}, to see that, for $\bQ, \bQ' \in \cQ$, our additional condition is 
\[ \frac{q_1}{q_2} \ind{q_1+q_2>0}= \frac{q_1'}{q_2'} \ind{q_1'+q_2'>0} \q{and} \frac{q_3}{q_4} \ind{q_3+q_4>0}= \frac{q_3'}{q_4'} \ind{q_3'+q_4'>0}. \]
Thus we see that any pasting $\bQ \oplus_{t=1} \bQ'$ that satisfies this condition is in fact equal to $\bQ$, which is trivially in $\cQ$.

\subsection{Reserving for cash flows}
We describe a probabilistic approach to wealth processes using the notation of Acciaio, F\"ollmer, and Penner \cite{AFP12}. As before, we fix a terminal time $T<\infty$, a discrete time set $\bT := \{0,1,\dots, T\} $, and a stochastic basis $(\Omega,\cF,(\cF_t)_{t\in \bT} , \P)$. On the product space  $\underline{\Omega} := \Omega \times \bT$, define the optional $\sigma$-algebra up to time $t \in \bT$ as
\begin{align*}
\underline{\cF}_t &:= \sigma\left( A_s \times \{s\}, A_t \times \bT_t: s \le t, A_s \in \cF_s\right), &\text{where }\bT_t := \{t, t+1,\dots, T\} ,
\\\underline{\cF} &:= \underline{\cF}_T.
\end{align*}
Define the reference probability measure $\underline{P}:= \P \otimes \mu$  on $(\underline{\Omega}, \underline{\cF})$ via the expectation 
\[ \bE_{\underline{P}} [X] = \bE\left[ \sum_{s=0}^T X_s \mu_s \right] \]
where $\bE=\bE_\P$ and $\mu$ is an \emph{optional random probability measure} on $\bT$, i.e., an $\cF_t$-adapted process such that $\mu_t >0$ for all $t \in \bT$ and $\sum_{t \in \bT} \mu_t=1$. 

We use the underline to denote multiperiod variants of standard notation; for example $\underline{L}^\infty := L^\infty (\underline{\Omega}, \underline{\cF}, \underline{P})$ is the space of all bounded random variables on the extended probability space $ (\underline{\Omega}, \underline{\cF}, \underline{P})$, elements of which may alternatively be viewed as processes $X = (X_t)_{t \in \bT}$. We write $\underline{\cL}^{1}(\R^{d+1}) := L^1 (\underline{\Omega}, \underline{\cF}, \underline{P}; \R^{d+1})$ (respectively $\underline{\cL}^{\infty}(\R^{d+1})$) for $\underline{P}$-integrable (respectively  bounded) random variables $X$ such that each $X_t$ is $\R^{d+1}$-valued, for $t\in \bT$. Non-negative elements of $\underline L^\infty $ are denoted by $\underline L^\infty_+$, and $\underline\cF_t$-measurable elements of $\underline L^\infty $ are denoted by $\underline L^\infty_t$.

For $0 \le t \le s \le T$, define the projection $\pi_{s,t}:\underline{L}^\infty \to \underline{L}^\infty$
\[ \pi_{s,t}(X)_r = \ind{s \le r} X_{r \wedge t} ,\quad \for r \in \bT. \]
Define  $\cR^\infty$ to be those \emph{adapted}  processes $X \in \underline{L}^\infty$, and set $\cR^\infty_{t,s}=\pi_{s,t}(\cR^\infty)$ and $\cR^\infty_{t}=\pi_{t,T}(\cR^\infty)$. We use the notation $X|_t$ for the conditional expectation $\bE_{\underline P}[X| \underline \cF_t]\equiv \bE[X| \underline \cF_t]$, which may be viewed as a process, constant after time $t$; we write $X_t$ to denote the time-$t$ realisation of the process $X$.

We remark that there is a one-to-one correspondence between pricing measures for processes $\rho_t : \cR^\infty_t \to L^\infty_t$ and pricing measures $\underline \rho_t: \cR^\infty \to \underline L^\infty_t $ for random variables on $\underline \Omega$ equipped with the optional $\sigma$-algebra, via
\begin{equation}
\label{eq:ccrmrv} \underline \rho_t(X) = \sum_{s=0}^{t-1} X_s \ind{s} + \rho_t(\pi_{t,T}(X))\one{\bT_t}.
\end{equation}

\section{Proof of main result}
\label{sec:pomr}
The proof is a little involved. We show that $(i)\Rightarrow(ii)\Leftrightarrow(iii)$. Then noting that (see Remark \ref{time}) (iii) implies the same result on the filtration $(\cF_t,\ldots,\cF_T)$, we deduce a time-shifted version of (ii) which implies (i).

We will use the following lemma:
\begin{lem}\label{lem:coneflip}
Suppose for each $t \in \bT$, $\cC_t \subset E$ is a closed convex cone. Then
\[ \left( \cap_t \cC_t\right)^* = \ol{\conv\left\{\cup_t \cC_t^*\right\}} = \ol{\oplus_t \cC_t^*}. \]
\end{lem}
\begin{proof}
The first equality is well-known, the second is clear. 
\end{proof}

\begin{proof}[Proof of \Cref{thm:main}]
$(i)\Rightarrow (ii)$

We show first that
\begin{claim}
$$\ce_t(X) := \essinf \{ \rho_t(Y \cdot \V ) :\; Y\in L^\infty_{t+1}(\R^{d+1})\text{ and }X-Y\cdot \V\in \cA_{t+1}\}
$$
is a conditional coherent risk measure and its acceptance set $\cB_t$ is given by
\begin{equation}\label{induct}
\ol{K_t(\V)\cdot\V + \cA_{t+1}}.
\end{equation}
\end{claim}
Now (i) implies that $\cB_t=\cA_t$ and so, from this result, the implication $(i)\Rightarrow (ii)$ is an immediate consequence by induction.

Note that, by a similar argument, it is easy to see that (i) is {\em equivalent} to the statement
\begin{equation}\label{unif}
\text{For each }t,\; \cA_t(\V)=\ol{\oplus_t^{T-1}K_t(\cA_0,\V)}.
\end{equation}

It is easy to show that $\ce_t$ is a conditional coherent risk measure. It remains
to show (\ref{induct}).
For the  inclusion $\cB_t\subseteq \ol{K_t(\V)\cdot \V + \cA_{t+1}}$, let $X\in \cB_t$. We first prove that the set $S_t(X):= \{ \rho_t(Y\cdot \V ) :\; Y\in L^\infty_{t+1}(\R^{d+1})\text{ and }X-Y\cdot \V\in \cA_{t+1}\}$ is directed downwards. To see this, take $Y, Z\in L^\infty_{t+1}(\R^{d+1})$ such that $X-Y\cdot \V\in \cA_{t+1}$ and $X-Z\cdot \V\in \cA_{t+1}$. Let $F=\{ \rho_t(Y\cdot \V )\leq \rho_t(Z\cdot \V)\}$ and define $W=X1_F+Y1_{F^c}$. Then $W\in  L^\infty_{t+1}(\R^{d+1})$ and $X-W\cdot \V\in \cA_{t+1}$. Then $\rho_t(W\cdot \V)=\rho_t(Y\cdot \V)1_F+\rho_t(Z\cdot \V)1_{F^c}=\min(\rho_t(Y\cdot \V),\rho_t(Z\cdot \V))$.

   Thus there exists a sequence $Y_n\in S_t(X)$ such that
$a_n:= \rho_t(Y_n\cdot \V )\downarrow \ce_t(X)$. Setting $U_n := Y_n+(\ce_t(X)-a_n)e_0$, where $e_0:=(1,0,\ldots, 0)$, it is clear that $U_n\in K_t(\V)\cdot \V$. Defining $X_n:=U_n\cdot \V+(X-Y_n\cdot \V)$ we see that $X_n\in K_t(\V) \cdot \V+\cA_{t+1}$ and the sequence $X_n = X+\ce_t(X)-a_n$
is uniformly bounded and converges a.s. to $X$. Therefore $X_n$ converges weakly$^*$ to X
and so $X\in \ol{K_t(\V)\cdot \V + \cA_{t+1}}$.

To prove the inclusion $\cB_t\supseteq \ol{K_t(\V) + \cA_{t+1}}$,
suppose that $X\in \cA_{t+1}$ and $Y\in K_t(\V)$.

  Then clearly $\rho_t(Y\cdot \V)\in S_t(Y\cdot \V)$ and so  $\ce_t(Y\cdot \V )\leq\rho_t(Y\cdot \V )\leq 0$, and for
$X\in\cA_{t+1}$ we have  $\ce_t(X)\leq \rho_t(0) = 0$, so $\ce_t(X+Y\cdot \V)\leq \ce_t(X)+\ce_t(Y\cdot \V)\leq 0$ so that $X+Y\cdot \V\in \cB_t$. Since $\cB_t$ is weak$^*$-closed, the result follows.

{\em Equivalence of (ii) and (iii)}

 Assume now that $\uB$ is a weak$^*$-closed convex cone in $\ul\cL^\infty(\R^{d+1})$ which is arbitrage-free, so that $\uB^{**} = \uB$. Define 
 \[ K_t(\uB) := \{ X \in\ul \cL^\infty(\cF_{t+1},\R^{d+1}): \alpha X \in \uB \text{ for any } \alpha \in \ul L^\infty_+(\cF_t) \}.\]
Recall that $\uB_t= \{ X \in\ul \cL^\infty(\cF_{T},\R^{d+1}): \alpha X \in \uB \text{ for any } \alpha \in \ul L^\infty_+(\cF_t) \}$ and $\uB$ is predictably representable if 
\[\uB = \ol{\oplus_{t=0}^{T-1} K_s(\uB)}.\]
 We rephrase (ii) and (iii) as follows
\begin{enumerate}[(i')]
\setcounter{enumi}{1}
\item
$\uB$ is predictably representable; and
\item
$\uB^*$ is predictably stable.
\end{enumerate}
This is sufficent, since we apply it to the case where $\uB=\cA_0(\mathbf{V})$.

\emph{(ii') $\Rightarrow$ (iii'):} Assuming $\uB$ is predictably representable, it follows from \Cref{thm:crucialClaim} that
\[ \uB = \ol{\oplus_{t=0}^{T-1} K_s(\uB)}^{w^*} =  \ol{\oplus_{t=0}^{T-1} \cM_s(\uB^*)^*}^{w^*}.\]
Taking the dual, we find that
\[   \uB^* = \cap_{s=t}^{T-1} \cM_s(\uB^*)^{**} = \cap_{t=0}^{T-1} \ol{\conv}\cM_s(\uB^*)  \]
where the last equality follows from the Bipolar Theorem for a locally convex topological vector space (see  the Fenchel-Moreau duality Theorem in \cite{R70}). Hence, $ \uB^* = [\uB^*]$, and by \Cref{lem:A2}, $\uB^*$ is predictably stable.

\emph{(iii') $\Rightarrow$ (ii'):} Assuming $\uB$ is a weak$^*$-closed convex  cone, note that $\uB^*$ is a convex cone closed in $(\ul\cL^1, \sigma(\ul\cL^1, \ul\cL^\infty))$. Assuming further that $\uB^*$ is \emph{stable},
\begin{align*}
\cB^* &= \cap_t \cM_t(\cB^*) &\text{by \Cref{lem:A1}}
\\&=  \cap_t K_t(\cB)^* &\text{by \cref{eq:MandK}}.
\end{align*}
Now we may apply \Cref{lem:coneflip} to deduce
\[ \uB\equiv \uB^{**} = \ol{\oplus_{t=0}^{T-1} K_s(\uB)}^{w^*}\]
and $\uB$ is predictably representable, as required.

Now we show that (iii) implies the uniform, time-shifted, version of (ii)':
$$
\cA_t(\V)=\ol{\oplus_{s=t}^{T-1}K_s(\cA_0,\V)},
$$
which is sufficient for (i) to hold. But this follows immediately from the observation that (iii)' implies the same condition holds for the filtration $(\cF_t,\ldots, \cF_T)$.
\end{proof}
\vfill\eject
All that remains is to give the
\begin{proof}[Proof of \Cref{thm:crucialClaim}] We set $\cB =\cAn(\mathbf{V})$,   as above. 

First we prove that $\cM_t(\cB^*) \subset K_t(\cB)^*$. For arbitrary $ Z\in \cM_t(\cB^*)$, there exist $Z'\in \cB^*$ and $\alpha \in\ul L^0_+(\cF_t)$ with $\alpha Z' \in\ul \cL^1$ and $Z|_{t+1}=\alpha Z'|_{t+1}$. 

Note that, for any $X \in K_t(\cB)$,
\[ \bE[Z\cdot X]=\bE[Z|_{t+1}\cdot X] =\bE[\alpha Z'|_{t+1}\cdot X] = \lim_{n \to \infty} \bE[(\alpha \ind{\alpha \le n}X)\cdot Z'|_{t+1}] \le 0,\]
since $\alpha \ind{\alpha \le n} X\in \cB$ and $Z'\in \cB^*$. Hence $Z \in K_t(\cB)$, and since $Z$ is arbitrary, we have shown that $\cM_t(\cB^*) \subset K_t(\cB)^*$. 

For the reverse inclusion, $\cM_t(\cB^*)^* \subset K_t(\cB)$, note that $\cB^* \subset \cM_t(\cB^*) $ implies $ \cM_t(\cB^*)^*  \subset \cB$, and 
\begin{align*}
\ul{\cL}^\infty_+(\ul{\cF}_t) \cM_s(\cD ) =\cM_s(\cD ) &\qiq  \for X\in \cM_t(\cB^*)^*, \quad  g \in \ul{\cL}^\infty_+(\ul{\cF}_t),\quad \bE[X\cdot gZ] \le 0
\\& \qiq \ul{\cL}^\infty_+(\ul{\cF}_t) \cM_t(\cB^*)^* = \cM_t(\cB^*)^*.
\end{align*}
Define \[\cB_t:= \{ X \in\ul \cL^\infty(\ul\cF_T,\R^{d+1}): g X \in \cB \text{ for any } g \in\ul L^\infty_+(\ul\cF_t) \}.\]
Thus $\cM_t(\cB^*)^* \subseteq \cB_t$.
To finish the proof, we need only show that $X \in \cM_t(\cB^*)^*$ is $\cF_{t+1}$-measurable, since $\cB_t \cap\ul \cL^\infty(\ul\cF_{t+1},\R^{d+1}) = K_t(\cB)$.

To this end, note that for any $Z \in\ul \cL^1(\R^{d+1})$, it is true that $Z-Z|_{t+1} \in \cM_t(\cB^*)$, whence $\bE[(Z-Z|_{t+1})\cdot X] \le 0$. We deduce that 
\[ \bE[(Z-Z|_{t+1})\cdot X] = \bE[(X-X|_{t+1})\cdot Z] \le 0 \qquad \forall Z \in \ul\cL^1(\R^{d+1}),\]
and $X= X|_{t+1}$ $\P$-a.s..
\end{proof}

\bibliographystyle{plain}	
\bibliography{bibliog}


\end{document}